\documentclass[submission,copyright,creativecommons]{eptcs}
\providecommand{\event}{GandALF 2012} 
\usepackage{breakurl}

\usepackage{amssymb}
\usepackage{amsmath}
\usepackage{amsthm}
\newtheorem{theorem}{Theorem}
\newtheorem{definition}[theorem]{Definition}
\newtheorem{proposition}[theorem]{Proposition}
\newtheorem{remark}[theorem]{Remark}
\newtheorem{example}[theorem]{Example}

\title{A Myhill-Nerode theorem for automata with advice}
\author{Alex Kruckman
\institute{University of California, Berkeley}
\email{kruckman@gmail.com}
\and
Sasha Rubin
\institute{TU Vienna and IST Austria}
\email{sasha.rubin@gmail.com}
\and
John Sheridan
\email{jhs223@cornell.edu}
\and
Ben Zax
\email{benzax@gmail.com}
}

\def\titlerunning{A Myhill-Nerode theorem for automata with advice}
\def\authorrunning{A. Kruckman, S. Rubin, J. Sheridan \& B. Zax}

\begin{document}

\maketitle

\begin{abstract}
An automaton with advice is a finite state automaton which has access to an additional fixed infinite string called {\em an advice tape}. We refine the Myhill-Nerode theorem to characterize the languages of finite strings that are accepted by automata with advice. We do the same for tree automata with advice.
\end{abstract}

\section{Introduction} \label{sec:intro}

%
%
Consider an extension of the classical model of finite automata operating on finite strings in which the machine simultaneously reads a fixed {\em advice tape} --- an infinite string $A$.  A {\em deterministic finite-string automaton with advice $A$} is like a deterministic finite-string automaton, except that at step $n$ the next state depends on the current state, the $n^{\text{th}}$ symbol of the input $w \in \Sigma^*$, and the $n^{\text{th}}$ symbol of $A \in \Gamma^\omega$. The automaton halts once all of $w$ has been read and accepts $w$ if and only if it is in a final state (in the {\em non-terminating} model of Section~\ref{sec:nonterminating}, the automaton reads the rest of the advice tape and accepts according to a Muller condition).

We now give a formal definition. 

\begin{definition} \label{dfn_autadvice}
An automaton with advice is a tuple $M = (Q,\Sigma,\Gamma,A,\delta,q_0,F)$.
\begin{enumerate}
\item $Q$ is the finite set of \emph{states} of the automaton.
\item $\Sigma$ is a finite set of symbols called the \emph{input alphabet}.
\item $\Gamma$ is a finite set of symbols called the \emph{advice alphabet}.
\item $A \in \Gamma^{\omega}$ is the \emph{advice string}.
\item $\delta: Q\times\Gamma\times\Sigma\to Q$ is the \emph{transition function}. 
\item $q_0\in Q$ is the \emph{initial state}.
\item $F \subseteq Q$ is the \emph{acceptance condition}.
\end{enumerate}

The \emph{run} of $M$ on a string $w\in\Sigma^*$ is a sequence of states $\alpha\in Q^{|w|+1}$ such that $\alpha_0 = q_0$ and for $1 \leq n \leq |w|$, $\alpha_n = \delta(\alpha_{n-1},A_n,w_n)$.

We say a string $w\in\Sigma^*$ is \emph{accepted} by the automaton with advice $M$ if the final state appearing the run of $M$ on $w$ is in $F$.

A language $L \subseteq \Sigma^*$ is {\em regular with advice $A$} if it is the language accepted by some automaton with advice $A$. A language $L$ is {\em regular with advice} if there exists $A$ such that $L$ is regular with advice $A$. Thus $L$  is {\em not regular with advice} means that there is no  $A$ such that $L$ is regular with advice $A$.
\end{definition}


What is the power of this model of computation? We make some trivial observations:
\begin{enumerate}
\item Every regular language is regular with advice (indeed, any advice will do).
\item Every unary language $L$ (ie. subset of $\{1\}^*$) is regular with advice (indeed, let $A$ be the characteristic sequence of $L$). 
\end{enumerate}
Is every language regular with advice? For a fixed advice $A$ the answer is clearly `no': there are continuum many languages and only countably many languages regular with advice $A$ (since there are countably many finite automata).

This simple argument does not preclude the possibility that for every language $L$ there is an advice $A_L$ such that $L$ is regular with advice $A_L$. Also note that the standard pumping and Myhill-Nerode arguments showing non-regularity do not apply in the presence of advice. We now sketch an alternative argument that in fact the language $\{0^n1^n : n \in \mathbb{N}\}$ is not regular with advice.

Recall that for a language $L$, the equivalence relation $\equiv_L$ on $\Sigma^*$, called the Myhill-Nerode congruence, is defined as follows: $x \equiv_L y$ if for all $z \in \Sigma^*$ it holds that $xz \in L \iff yz \in L$.
The classical Myhill-Nerode theorem states:

\begin{theorem}[Myhill-Nerode]
A language $L \subseteq \Sigma^*$ is regular if and only if $\equiv_{L}$ has finitely many equivalence classes.
\end{theorem}

The proof of the Myhill-Nerode theorem for classical automata suggests the following observation regarding automata with advice:

Let $M$ be an automaton with advice which accepts the language $L$ with some advice $A$. Suppose that $M$, starting in the initial state, reaches the same state on input $x$ as on input $y$. If $M$ is also at the same place in the advice tape $A$ after reading $x$ and $y$, that is, if $x$ and $y$ have the same length, then $x \equiv_L y$. Thus, for every $n$, the number of classes of $\equiv_{L}$ restricted to $\Sigma^n$ (the strings in $\Sigma^*$ of length exactly $n$) is at most the number of states in $M$.

For language $L$ and integer $n$ write $\equiv_{L,n}$ for the equivalence relation $\equiv_L$ restricted to $\Sigma^n$.  If there exists $n$ such that the number of equivalence classes of $\equiv_{L,n}$ is $k$, then no automaton with advice having fewer than $k$ states accepts $L$.

We now have a way to prove that certain languages are not regular with advice. Consider $L :=\{0^n1^n : n \in \mathbb{N}\}$ and note that for every $n$, no pair of strings in the set
$X_n := \{0^a 1^{n-a} : \frac{n}{2} \leq a \leq n\}$ are $\equiv_{L,n}$-equivalent. But the size of $X_n$ is unbounded as $n$ grows. Thus $L$ is not regular with advice.

The observation above gives one direction of a Myhill-Nerode like characterization. We prove the other direction in Section \ref{sec:m-n}. The role of the size of the alphabet $\Gamma$ is considered in Section \ref{sec:alphabet}. A variation of the model --- which we call {\em non-terminating} ---  in which the automaton reads the rest of the advice is considered in Section
\ref{sec:nonterminating}. 
Finally, we mention in Section~\ref{sec:trees} that the results go through for tree automata with a fixed infinite tree as advice.

\subsection*{Related work}

%

Automata over finite words can be identified with weak monadic second-order formulas over the structure $(\mathbb{N},\mathrm{succ})$. Non-terminating automata with advice correspond to WMSO formulas over expansions of $(\mathbb{N},\mathrm{succ})$ by unary predicates $\bar{P}$ (this is implicit in \cite{ElRa66,CaTh02,RaTh06, Bara06Hierarchy}). Questions of logical decidability are equivalent to the $\bar{P}$-acceptance problem: given a Muller automaton $M$, decide whether or not $M$ accepts $\bar{P}$. Similar things are done for automata operating on finite trees \cite{Fratani}[Definition $11$] and \cite{CoLo07}.



\def\Q{\mathbb{Q}}

It is easy to see that the languages recognized by non-terminating automata with advice are closed under logical operations such as union, complementation, projection, permutation of co-ordinates and instantiation (see for instance \cite{Fratani}). Consequently, one may define {\em automatic structures with advice} (\cite{CoLo07}). These are relational structures whose domain and atomic relations are recognized by automata with advice. The case without advice is well studied, and such structures have decidable first-orer theory (see \cite{Rubi08}). The main programme there has been to supply techniques for showing non-automaticity.
For instance, there is a difficult proof of the fact that $(\Q,+)$ is not automatic without advice (\cite{Tsan11}). However it is automatic with advice (communicated by Frank Stephan and Joe Miller, and reported in \cite{Nies07}). Because this example is not yet well known we present it here.

\begin{example}
To simplify exposition, we give a presentation of $([0,1) \cap \Q,+)$. Each rational is coded by a finite string
over the alphabet $\{0,1,\#\}$. Automata for the domain and the addition will have access to the advice string
\[
A = 10\#11\#100\#101\#110\#111\#1000\#\cdots
\]
which is a version of the Champernowne-Smarandache string. 
To every rational $q$ in $[0,1)$ there is a unique {\em finite}
sequence of integers  $a_1 \cdots a_n$ such that $0 \leq a_i < i$, $q = \sum_{i=2}^n \frac{a_i}{i!}$, and $n$ is minimal. The presentation codes this rational as
$f(a_2)\#f(a_3)\#f(a_4) \cdots \#f(a_n)$ where $f$ sends $a_i$ to the binary string of length $\lceil  \log_2 i \rceil +1$  representing $a_i$. Addition $a+b$ is performed least significant digit first (right to left) based on the fact that 
\[
\frac{a_i + b_i + c}{i!} = \frac{1}{(i-1)!} + \frac{a_i + b_i + c - i}{i!}
\]
where $c \in \{0,1\}$ is the carry in. In other words, if $a_i+b_i+c \geq i$ then write $a_i+b_i+c-i$ in the $i$th segment and carry a $1$ into the $(i-1)$st segment; and if $a_i +b_i + c < i$ then write this under the $i$th segment and carry a $0$ into the $(i-1)$st segment.  These comparisons and additions can be performed since the advice tape is storing $i$ in the same segment as $a_i$ and $b_i$. 
\end{example}

We remark that the advice string $A$ above has decidable acceptance problem (\cite{Bara06Hierarchy}). Consequently every structure that is automatic with this advice has decidable first-order theory.

We end with a question that we hope will spur interest: what are other interesting examples of structures that are automatic with advice?

\subsection*{Acknowledgements} 
Research supported by the National Science Foundation through the Research Experiences for Undergraduates (REU) program at Cornell University, the European Science Foundation for the activity entitled `Games for Design and Verification', FWF NFN Nr. S11407-N23 (RiSE), ERC 279307: Graph Games, PROSEED WWTF Nr. ICT 10-050, and  ARISE FWF Nr. S11403-N23.

\section{The Myhill-Nerode Theorem} \label{sec:m-n}

Let $\Sigma$ be a finite alphabet and $L\subseteq \Sigma^*$ a language. Define an equivalence relation $\equiv_{L,n}$ on $\Sigma^n$, the set of all strings of length $n$, by $x\equiv_{L,n} y$ if for all $z \in \Sigma^*$ it holds that $xz \in L \iff yz \in L$.
This is the usual Myhill-Nerode congruence restricted to strings of length $n$.

\begin{theorem}[Myhill-Nerode theorem with advice]
A language $L \subseteq \Sigma^*$ is regular with advice if and only if there is some 
$k \in \mathbb{N}$ such that for every $n \geq 0$, $\equiv_{L,n}$ has at most $k$ equivalence classes.
\end{theorem}

\begin{proof}
Suppose $L$ is regular with advice. Let $M=(Q,\Sigma,\Gamma,A,\delta,q_0,F)$ be the automaton recognizing $L$, and let $k = |Q|$. Now assume
that for some $n$, $\equiv_{L,n}$ divides the strings of length $n$ into $l$
equivalence classes, with $l>k$. Pick representative strings $x_1,\hdots,x_l$ in these classes. Let $q_i\in Q$ be the $(n+1)^{\text{st}}$ state in the run of $M$ on
$x_i$, that is, the state reached after reading the final character of $x_i$. Since $l > |Q|$, $q_i = q_j$ for some $i\neq j$. Then
for all $z\in\Sigma^*$, $M$ accepts $x_iz$ if and only if $M$ accepts $x_jz$,
since the run of $M$ on these two strings is the same after stage $n$. This contradicts the assumption that $x_i$ and $x_j$ are representatives of distinct $\equiv_{L,n}$- classes.

Conversely, suppose we have such a bound $k$. For each $n\in \mathbb{N}$, let $\mathcal{C}_n$ be the collection of equivalence classes of $\equiv_{L,n}$, so $|\mathcal{C}_n|\leq k$. 

We must construct an automaton $M$ recognizing $L$. Let $Q = \{1,\dots,2k\}$ be the set of states, and let $F = \{1,\dots,k\}\subset Q$. The idea is that we have enough room to represent each equivalence class by an accepting or a rejecting state as necessary at each stage. Set $q_0 = 1$ if the empty string is in $L$ and $q_0 = k+1$ otherwise. 

For $x,y\in \Sigma^n$, if $x\equiv_{L,n}y$, then $x\in L$ if and only if $y\in L$ (appending the empty
string as a suffix). Thus, $\mathcal{C}_n$ is partitioned into those classes which are  ``accepting'' and those which are not. Also if $x\equiv_{L,n} y$, then for all $a\in\Sigma$, $xa\equiv_{L,n+1}ya$, since for all $z\in\Sigma^*$,
$xaz\in L$ if and only if $yaz\in L$. This defines a function
$h_n:\mathcal{C}_n\times\Sigma\rightarrow\mathcal{C}_{n+1}$ so that if $C\in \mathcal{C}_n$ and $x$ is a string in $C$, then $xa$ is a string in the class $h_n(C,a)$. 

For each $n>0$, identify the ``accepting'' classes with states from $\{1,\dots,k\}$ and the remaining classes with states from $\{k+1,\dots,2k\}$. The remaining work is to encode the transition information given by the functions $h_n$ into the advice tape.

Enumerate all functions $Q\times \Sigma \rightarrow Q$ by $\langle f_i\rangle_{i = 1}^N$, where $N = (2k)^{2k|\Sigma|}$, and let $\Gamma = \langle c_i\rangle_{i = 1}^N$ be the advice alphabet. Each character codes a possible transition behavior. For each $n\in\mathbb{N}$, pick a function $f_i$ which respects $h_n$ in the
sense that if a class $C\in \mathcal{C}_n$ is associated to the state $j$, then for any character $a\in \Sigma$, $f_i(j,a)$ is the state associated to $h_n(C,a)$. Since not every state is associated to a class, $f_i$ may behave arbitrarily on some inputs. Set the $n^{\text{th}}$ character of the advice tape $A$ to be $c_i$. 

Finally, we define the transition function $\delta : Q\times \Gamma\times \Sigma\rightarrow Q$ by
$\delta(j,c_i,a) = f_i(j,a)$. It is easy to check by induction on length that $M$ accepts the string $x$ if and only if $x\in L$.
\end{proof}


\section{On the role of alphabet size}\label{sec:alphabet}
In our proof of the Myhill-Nerode theorem with advice, we made use of a large advice alphabet. This raises the question of whether the size of the advice alphabet is essential.

\begin{quote}
Does there exist $k$ such that if $L$ is regular with advice, then already $L$ is regular with some advice over an alphabet of size $k$?
\end{quote}

The answer is `no'. For $k \in \mathbb{N}$ let $\mathrm{REGA}_k$ be the set of languages that are regular in some advice with advice alphabet of size $k$. Then $\mathrm{REGA}_1$ are the regular languages, and $\mathrm{REGA}_k \subseteq \mathrm{REGA}_{k+1}$. We prove that
$\mathrm{REGA}_k \neq \mathrm{REGA}_{k+1}$ for all $k \in \mathbb{N}$.

For $A \in \{0,\dots,k\}^\omega$, let $\text{Pref}(A)$ be the language consisting of all prefixes (initial segments) of $A$. The set $\text{Pref}(A)$ is clearly regular with advice $A$, but if we choose $A$ carefully, then $\text{Pref}(A)$ is not regular with any advice $B\in \{0,\dots,k-1\}^\omega$.

%
%

To simplify the proof, we will change the question to one about deterministic transducers. A \emph{deterministic transducer} is a machine $M$ which reads an infinite input string $B$ and produces an infinite output string $A$. We denote this by $M[B] = A$. At each stage, $M$ produces an output character based on an input character and its current state.

If we have an automaton $M$ recognizing $\text{Pref}(A)$ with advice $B$, we can transform it into a deterministic transducer $M^*$ such that $M^*[B] = A$. Note that on the run of $M$ on the (infinite) input string $A$, $M$ is always in an accepting state, and for each state and advice character pair $(q,b)$ occuring in the run, there is exactly one input character $a$ (the next character of $A$) such that $\delta(q,b,a)$ is an accepting state. Without changing the language accepted by $M$, we may adjust the transition function so that for each pair of an accepting state $q\in F$ and an advice character $b$ (even those pairs not appearing in the run of $M$ on $A$), there is exactly one input character $a$ such that $\delta(q,b,a)$ is an accepting state. 

Now we define the deterministic transducer $M^*$ so that it reads the advice string and at each stage produces the unique acceptable input character as output. That is, $M^*$ has the same state set as $M$, and if $M^*$ reads the advice character $b$ in state $q$, it produces the output character $a$ as above and transitions to $\delta(q,b,a)$.

%
%
\begin{proposition}
There exists an infinite string $A$ over alphabet $\Sigma = \{0, \dots, k\}$ such that for every deterministic transducer $M$ (with input alphabet $\Gamma = \{0, \dots, k-1\}$ and output alphabet $\Sigma$) and every infinite string $B$ over $\Gamma$, $M[B] \neq A$.
\end{proposition}

\begin{proof}
We diagonalize. For each deterministic transducer $M$ with input alphabet $\Gamma$ and output alphabet $\Sigma$, we will produce a finite string $u_M$ from $\Sigma$ such that $u_M$ does not appear as a substring of $M[B]$ for any string $B\in\Gamma^\omega$.

Let $Q$ be the states of $M$. For $q\in Q$ and $b\in \Gamma$, write $M^q[b]$ for the character in $\Sigma$ produced by $M$ in state $q$ upon reading the character $b$. Among the $k|Q|$ characters $M^q[b]$ (parametrized by $b \in \Gamma$ and $q \in Q$), there must be a character, say $a \in \Sigma$, that occurs at most $\frac{k}{k+1}|Q|$ times. Let $a$ be the first character of $u_M$. Let $Q'$ be those states $q'\in Q$ such that for some $b \in \Gamma$ and $q \in Q$, $M^q[b] = a$ and $M$ transitions to $q'$. Note that $|Q'| \leq \frac{k}{k+1}|Q| < |Q|$, so $Q'$ is a proper subset of $Q$. Replace $Q$ by $Q'$ and repeat to get the next character of $u_M$ and a new set of states $Q''\subset Q'$. After a finite number of steps, we will have $|Q^{(n)}| = 0$, at which point we have finished constructing $u_M$. 

Now we claim that $u_M$ cannot appear as a substring of $M[B]$ for any $B\in \Gamma^\omega$. Suppose for contradiction that it does appear in the run of $M$ on $B$. At the step in which $M$ begins producing $u_M$, it is in some state $q\in Q$. It reads a character from $B$, outputs the first character of $u_M$, and transitions into a state in $Q'$. After the next step, $M$ transitions into a state in $Q''$. After $n$ steps, regardless of the  content of $B$, $M$ must be in a state in $Q^{(n)}$. But $Q^{(n)}$ is empty.

Concatenating the countably many strings $u_M$, we obtain our string $A$. For every deterministic transducer $M$ and every infinite string $B$, $M[B]\neq A$, witnessed by the presence of $u_M$ as a substring of $A$.
\end{proof}

\section{Nonterminating automata with advice} \label{sec:nonterminating}

In the introduction we defined an automaton with advice as one that terminates after reading its finite input string (Definition~\ref{dfn_autadvice}). In this section, if $L$ is regular with advice $A$ we will say that $L$ is  {\em terminating regular with advice $A$}. There is another way to define automaton with advice, in what we call the non-terminating model. Here automata are defined with an infinitary acceptance condition. After the input string ends, the automaton continues to read the advice, producing an infinite run. The run is successful if it satifies the acceptance condition. We work with deterministic Muller automata, in which the input string is accepted if the set of states visited infinitely often is in the collection $F$ of accepting sets of states. 

\begin{definition}
A nonterminating automaton with advice is a tuple $M = (Q,\Sigma,\Gamma,A,\delta,q_0,F)$, with data just as before, except: 
\begin{itemize} 
\item $\delta: Q\times\Gamma\times(\Sigma\cup\{\Box\})\to Q$ is the \emph{transition function}. Here blank ($\Box$) is a new symbol that is read once the input string has ended.
\item $F \subseteq \mathcal{P}(Q)$ is the \emph{acceptance condition}, where $\mathcal{P}(Q)$ is the powerset of $Q$.
\end{itemize}

The \emph{run} of $M$ on a string $w\in\Sigma^*$ is an infinite sequence of states $\alpha\in Q^{\omega}$ such that $\alpha_0 = q_0$ and:
\begin{itemize}
\item For $1 \leq n\leq|w|$: $\alpha_{n} = \delta(\alpha_{n-1},A_n,w_n)$.
\item For $n>|w|$: $\alpha_{n} = \delta(\alpha_{n-1},A_n,\Box)$.
\end{itemize}
We say a string $w\in\Sigma^*$ is \emph{accepted} by an automaton with advice $M$ if the set of states that appear infinitely often in $\alpha$ is an element of the acceptance condition $F$.

A language $L$ is {\em non-terminating regular with advice $A$} if it is the language accepted by some non-terminating automaton with advice $A$. A language $L$ is {\em non-terminating regular with advice} if there exists $A$ such that $L$ is non-terminating regular with advice $A$.
\end{definition}

The distinction between the terminating and nonterminating models was irrelevant for our discussion of Myhill-Nerode because these models are equivalent if we do not fix the advice:

\begin{proposition} \label{prop:eqmods}
For every language $L$, there exists advice $A$ such that $L$ is terminating regular with advice $A$ if and only if there exists advice $B$ such that $L$ is non-terminating regular with advice $B$.
\end{proposition}
\begin{proof}
Given a terminating regular language, it is clearly also nonterminating regular with the same advice: the automaton can simply ignore the remainder of the advice, looping forever on the final state.

In the other direction, suppose $M = (Q,\Sigma,\Gamma,A,\delta,q_0,F)$ is a non-terminating automaton accepting $L$. Let $\alpha$ be the run of $M$ on input $w$. Whether or not $\alpha$ is successful depends on
whether or not $M$ starting in state $\alpha_{|w|}$ and at position $n+1$ of the advice tape accepts the empty string.  This  information can be encoded in the advice if we expand the advice alphabet.  Formally, let $B$ be the infinite string whose $n^{th}$ letter is the pair $(A_n,f_n)$, where $A_n$ is the $n^{th}$ letter of the original advice $A$ and $f_n \subseteq Q$ consists of all states $q$ such that automaton $(Q,\Sigma,\Gamma,A[n+1,\infty),\delta,q,F)$ starting in state $q$ and position $n+1$ in the adivce $A$ accepts the empty string. Then $L$ is terminating regular with advice $A$.
%
\end{proof}

If $L_T(A)$ is the class of languages terminating regular with advice $A$ and $L_N(A)$ is the class of languages nonterminating-regular with advice $A$, then $L_T(A) \subseteq L_N(A)$.
However, the models are not equivalent for certain advice strings:

\begin{proposition}\label{prop:TvsNT}
There exists advice $A$ such that $L_T(A) \neq L_N(A)$.
\end{proposition}
\begin{proof}
Let $A$ be an infinite string on the binary alphabet $\Gamma = \{0,1\}$ such that $\text{Pref}(A)$ is not regular (without advice). Let $L$ be the language on the unary alphabet $\Sigma = \{0\}$ consisting of those strings of length $n$ such that the $(n+1)^{\text{st}}$ character of $A$ is a $1$. 

The language $L$ is easily seen to be nonterminating regular with advice $A$. Indeed, let $M$ be a machine which loops until it reads a blank. At that point, if the advice character is a $1$, it transitions to an infinite loop at a state $q$ with $\{q\} \in F$. If not, it transitions to an infinite loop at a state $q'$ with $\{q'\} \not \in F$.

But $L$ is not terminating regular with advice $A$. The intuition is that when the input string ends, the automaton cannot guess the next advice character. Suppose for contradiction that $M = (Q,\Sigma,\Gamma,A,\delta,q_0,F)$ is a machine in the terminating model recognizing $L$ with advice $A$. Then we can construct a machine $M' = (Q',\Gamma,\delta',q_0',F')$ recognizing $\text{Pref}(A)$ (without advice). 

Let $Q' = Q\cup\{r\}$, where $r$ is a single new state, $F' = Q\subset Q'$, and $q_0' = q_0$. Now define $\delta':Q'\times\Gamma\rightarrow Q'$ as follows:
$$\delta'(q,a) = \begin{cases}\delta(q,a,0) & \text{if } q\in F \text{ and } a = 1 \text{ or if } q \in Q\setminus F \text{ and } a = 0\\r & \text{if } q = r \text{ or if } q\in F\text{ and } a = 0 \text{ or if } q\in Q\setminus F\text{ and } a = 1 \end{cases}$$

If a state $q$ was an accepting state of $M$, it expects the next character of $A$ to be a $1$, and if $q$ was a rejecting state of $M$, it expects the next character of $A$ to be a $0$. Then $M'$ recognizes $\text{Pref}(A)$, contradicting our assumption on the complexity of $A$.
\end{proof}

\section{Tree automata with advice} \label{sec:trees}

The classical Myhill-Nerode theorem has a natural generalization to (leaf-to-root deterministic) tree automata \cite{Kozen92}. Our refinement can be easily adapted to this setting, giving a characterization of the sets of labeled trees which are regular with advice. 

\begin{definition} 
A \emph{(finite binary) tree} $t\subset \{0,1\}^*$ is a finite set of binary strings, called \emph{positions}, which is downwards closed, in the sense that if $w\in t$ and $v$ is an initial segment of $w$, then $v\in t$. A \emph{labeled tree} is a pair $(t,l)$, where $t$ is a tree and $l$ is a function $t\rightarrow \Sigma$ for some finite alphabet $\Sigma$.

The \emph{root} of a nonempty tree is the empty string $\lambda$. The \emph{children} of a position $w\in t$ are the positions $w0$ and $w1$.

Given a tree $t$ and a position $w\in t$ such that $w$ has a child $w'\notin t$, we call the child $w'$ a \emph{graft site}. We also consider the empty string $\lambda$ to be a graft site of the empty tree $\varepsilon$. If $u$ is a graft site of $t$, then for any other tree $x$ we define a new tree $t|_ux = t\cup\{uw\, |\, w\in x\}$. If $t$ and $x$ are labeled from $\Sigma$ by $l_t$ and $l_x$, then $t|_ux$ is also labeled from $\Sigma$: $l(w) = l_t(w)$ for $w\in t$ and $l(uw) = l_x(w)$ for $w\in x$.

An \emph{advice tree} is a labeling of the complete binary tree $\{0,1\}^*$ by a finite advice alphabet $\Gamma$, that is, a function $A:\{0,1\}^*\rightarrow \Gamma$ .

\end{definition}

A tree automaton operates on a labeled tree $t$ by assigning to all positions not in the tree a prescribed initial state. It then works inductively toward the root of the tree, assigning a state to each position in the tree based on the two states assigned to the children of that position and the label at that position. The automaton accepts if an accepting state is assigned to the root. A tree automaton with advice additionally has access to the advice character $A(u)$ when assigning a state to position $u \in t$. 

\begin{definition}
A tree automaton with advice is a tuple $M = (Q,\Sigma,\Gamma,A,\delta,q_0,F)$.
\begin{enumerate}
\item $Q$ is the finite set of \emph{states} of the automaton.
\item $\Sigma$ is a finite set of symbols called the \emph{input alphabet}.
\item $\Gamma$ is a finite set of symbols called the \emph{advice alphabet}.
\item $A:\{0,1\}^*\rightarrow \Gamma$ is the \emph{advice tree}.
\item $\delta: Q\times Q\times \Gamma\times\Sigma\to Q$ is the \emph{transition function}.
\item $q_0\in Q$ is the \emph{initial state}.
\item $F \subseteq Q$ is the \emph{acceptance condition}.
\end{enumerate}

The \emph{run} of $M$ on a labeled tree $(t,l)$ is an assignment $r:\{0,1\}^*\rightarrow Q$ of a state to each position in the complete binary tree such that if $w\notin t$, $r(w) = q_0$, and if $w\in t$, $r(w) = \delta(r(w0),r(w1),A(w),l(w))$. Since $t$ is finite, there is a unique such assignment. 

A labeled tree $(t,l)$ is \emph{accepted} by $M$ if $r(\lambda) \in F$. A set of labeled trees $T$ is {\em tree regular with advice} if it is the set accepted by some tree automaton with advice.
\end{definition}

\begin{remark}\label{rem_ntm}
We have just introduced a terminating model. Of course there is also a non-terminating model (in this case the automaton is non-deterministic, starts at the root, and a run is successful if every infinite path satisfies a Muller condition). Just as in string case
the two models are equivalent if we don't fix the advice (cf. Proposition~\ref{prop:eqmods}) and not necessarily equivalent if we do fix the advice (cf. Proposition~\ref{prop:TvsNT}).
\end{remark}

Let $\Sigma$ be a finite alphabet and $T$ a set of trees labeled from $\Sigma$. Define an equivalence relation $\equiv_T$ on the set of all trees labeled from $\Sigma$ by $x\equiv_T y$ if for any labeled tree $t$ and any graft site $u$ of $t$, it holds that $t|_ux \in T \iff t|_uy \in T$.

\begin{theorem}[Myhill-Nerode theorem for trees]  \cite{Kozen92}
A set of labeled trees $T$ is tree regular if and only if $\equiv_T$ has finitely many equivalence classes.
\end{theorem}

The definition must be modified in the presence of advice: Given a position $v\in \{0,1\}^*$, define the equivalence relation $\equiv_{T,v}$ on the set of all labeled trees by $x\equiv_{T,v} y$ if for any labeled tree $t$ such that $v$ is a graft site of $t$, it holds that $t|_vx \in T \iff t|_vy \in T$.

\begin{theorem}[Myhill-Nerode theorem for trees with advice]
A set of labeled trees $T$ is tree regular with advice if and only if there is some 
$k \in \mathbb{N}$ such that for all $v\in \{0,1\}^*$, $\equiv_{T,v}$ has at most $k$ equivalence classes.
\end{theorem}

The idea of the proof is the same as in the finite string case. If a set of labeled trees is regular with advice, the number of $\equiv_{T,v}$-classes must be bounded by the number of states. Conversely, given a uniform bound $k$, we may construct an automaton by associating a state to each equivalence class and encoding the transition information into the advice tree. 

We include the proof for completeness.

\begin{proof}
Suppose $T$ is tree regular with advice. Let $M=(Q,\Sigma,\Gamma,A,\delta,q_0,F)$ be the tree automaton recognizing $T$, and let $k = |Q|$. Now assume that for some position $v$, $\equiv_{T,v}$ divides the set of labeled trees into $n$ equivalence classes, with $n>k$. Pick representative labeled trees $x_1,\hdots,x_n$ in these classes. Let $t$ be a labeled tree with graft site $v$, and let $q_i\in Q$ be the state associated to the position $v$ (which is the root of $x_i$) in the run of $M$ on $t|_vx_i$. Note that this is independent of the choice of $t$. Since $n>|Q|$, $q_i = q_j$ for some $i\neq j$. Then for any tree $t$ with graft site $v$, $M$ accepts $t|_vx_i$ if and only if $M$ accepts $t|_vx_j$, since the states of $M$ assigned to positions in the base tree $t$ only depends on the states assigned to positions in the grafted tree based on which state is assigned to the string $v$. This contradicts the assumption that $x_i$ and $x_j$ are representatives of distinct equivalence classes.

Conversely, suppose we have such a bound $k$. For each position $v$, let $\mathcal{C}_v$ be the collection of equivalence classes of $\equiv_{L,v}$, so $|\mathcal{C}_v|\leq k$. Note that $\mathcal{C}_\lambda$ consists of at most two classes. Since the empty tree $\varepsilon$ is the only tree with graft site $\lambda$, and $\varepsilon |_\lambda t = t$, $s \equiv_{T,\lambda} t$ means that $s$ and $t$ are both in $T$ or both not in $T$.

We must construct an automaton $M$ recognizing $T$. Let $Q = \{1,\dots,k\}$ be the set of states, and let $F = \{1\}\subset Q$. Set $q_0 = 1$ if the empty tree is in $T$ and $q_0 = 2$ otherwise. For each position $v$, we will associate one state to each equivalence class in $\mathcal{C}_v$. This can be done arbitrarily, except for two requirements: 
\begin{itemize}
\item The state ($1$ or $2$) named $q_0$ must always be associated to the equivalence class containing the empty tree.
\item For $\mathcal{C}_\lambda$, the state $1$ must be associated to the class consisting of those trees in $T$.
\end{itemize}

For each input character $a\in \Sigma$, we may form the singleton labeled tree consisting of just the root position labeled by $a$: $t_a = (\{\lambda\}, \lambda\rightarrow a)$. Now $t_a$ has two graft sites, $0$ and $1$. If $s$ and $t$ are labeled trees, we can form the tree $t_{(s,a,t)} = (t_a|_0 s)|_1 t$. Now for any position $v$, if $s \equiv_{T,v0} s'$ and $t\equiv_{T,v1} t'$, then $t_{(s,a,t)} \equiv_{T,v} t_{(s',a,t')}$. Indeed, for any labeled tree $x$ with graft site $v$, \begin{eqnarray*}x|_vt_{(s,a,t)} \in T &\iff& ((x|_vt_a)|_{v0} s)|_{v1} t \in T\\ &\iff& ((x|_vt_a)|_{v0} s)|_{v1} t' \in T\\ &\iff& ((x|_vt_a)|_{v1} t')|_{v0} s \in T \\ &\iff& ((x|_vt_a)|_{v1} t')|_{v0} s' \in T\\ &\iff& x|_vt_{(s',a,t')} \in T.\end{eqnarray*}

This defines a function $h_v: \mathcal{C}_{v0} \times \mathcal{C}_{v1} \times \Sigma \rightarrow \mathcal{C}_v$ so that if $C\in \mathcal{C}_{v0}$, $C'\in \mathcal{C}_{v1}$ and $s\in C$, $t\in C'$, then $t_{(s,a,t)}$ is a member of the class $h_v(C,C',a)$. We now wish to encode this transition information into the advice tree.

Enumerate all functions $Q\times Q\times \Sigma \rightarrow Q$ by $\langle f_i\rangle _{i = 1}^N$, where $N = k^{k^2|\Sigma|}$, and let $\Gamma = \langle c_i\rangle_{i = 1}^N$ be the advice alphabet. Each character codes a possible transition behavior. For each position $v$, pick a function $f_i$ which respects $h_v$ in the sense that if states $j$ and $j'$ are associated to classes $C\in \mathcal{C}_{v0}$ and $C'\in \mathcal{C}_{v1}$, then for any character $a\in \Sigma$, $f_i(j,j',a)$ is the state associated to $h_n(C,C',a)$. Since not every state is associated to a class, $f_i$ may behave arbitrarily on some inputs. Set $A(v) = c_i$. 

Finally define the transition function $\delta:Q\times Q\times\Gamma\times \Sigma\rightarrow Q$ by $\delta(j,j',c_i,a) = f_i(j,j',a)$. It is easy to check that $M$ accepts the labeled tree $t$ if and only if $t\in T$.
\end{proof}

\bibliography{advicebib}{}

\begin{thebibliography}{10}
\providecommand{\bibitemdeclare}[2]{}
\providecommand{\surnamestart}{}
\providecommand{\surnameend}{}
\providecommand{\urlprefix}{Available at }
\providecommand{\url}[1]{\texttt{#1}}
\providecommand{\href}[2]{\texttt{#2}}
\providecommand{\urlalt}[2]{\href{#1}{#2}}
\providecommand{\doi}[1]{doi:\urlalt{http://dx.doi.org/#1}{#1}}
\providecommand{\bibinfo}[2]{#2}

\bibitemdeclare{inproceedings}{Bara06Hierarchy}
\bibitem{Bara06Hierarchy}
\bibinfo{author}{V.~\surnamestart B{\'a}r{\'a}ny\surnameend}
  (\bibinfo{year}{2006}): \emph{\bibinfo{title}{A Hierarchy of Automatic
  omega-Words having a Decidable MSO Theory}}.
\newblock In: {\sl \bibinfo{booktitle}{On-line proceedings of the 11th JournŽes
  Montoises, Rennes 2006}}.

\bibitemdeclare{article}{CaTh02}
\bibitem{CaTh02}
\bibinfo{author}{Olivier \surnamestart Carton\surnameend} \&
  \bibinfo{author}{Wolfgang \surnamestart Thomas\surnameend}
  (\bibinfo{year}{2002}): \emph{\bibinfo{title}{The monadic theory of morphic
  infinite words and generalizations}}.
\newblock {\sl \bibinfo{journal}{Inf. Comput.}}
  \bibinfo{volume}{176}(\bibinfo{number}{1}), pp. \bibinfo{pages}{51--65},
  \doi{10.1006/inco.2001.3139}.

\bibitemdeclare{article}{CoLo07}
\bibitem{CoLo07}
\bibinfo{author}{T.~\surnamestart Colcombet\surnameend} \&
  \bibinfo{author}{C.~\surnamestart L\"oding\surnameend}
  (\bibinfo{year}{2007}): \emph{\bibinfo{title}{Transforming structures by set
  interpretations}}.
\newblock {\sl \bibinfo{journal}{Logical Methods in Computer Science}}
  \bibinfo{volume}{3}(\bibinfo{number}{2}), \doi{10.2168/LMCS-3(2:4)2007}.

\bibitemdeclare{article}{ElRa66}
\bibitem{ElRa66}
\bibinfo{author}{C.C. \surnamestart Elgot\surnameend} \& \bibinfo{author}{M.O.
  \surnamestart Rabin\surnameend} (\bibinfo{year}{1966}):
  \emph{\bibinfo{title}{Decidability of extensions of theory of successor}}.
\newblock {\sl \bibinfo{journal}{J. Symb. Log.}}
  \bibinfo{volume}{31}(\bibinfo{number}{2}), pp. \bibinfo{pages}{169--181},
  \doi{10.2307/2269808}.

\bibitemdeclare{article}{Fratani}
\bibitem{Fratani}
\bibinfo{author}{S{\'e}verine \surnamestart Fratani\surnameend}
  (\bibinfo{year}{2012}): \emph{\bibinfo{title}{Regular sets over extended tree
  structures}}.
\newblock {\sl \bibinfo{journal}{Theoretical Computer Science}}
  \bibinfo{volume}{418}(\bibinfo{number}{0}), pp. \bibinfo{pages}{48 -- 70},
  \doi{10.1016/j.tcs.2011.10.020}.

\bibitemdeclare{article}{Kozen92}
\bibitem{Kozen92}
\bibinfo{author}{Dexter \surnamestart Kozen\surnameend} (\bibinfo{year}{1992}):
  \emph{\bibinfo{title}{On the {M}yhill-{N}erode theorem for trees}}.
\newblock {\sl \bibinfo{journal}{Bull. Europ. Assoc. Theor. Comput. Sci.}}
  \bibinfo{volume}{47}, pp. \bibinfo{pages}{170--173}.

\bibitemdeclare{article}{Nies07}
\bibitem{Nies07}
\bibinfo{author}{Andr{\'e} \surnamestart Nies\surnameend}
  (\bibinfo{year}{2007}): \emph{\bibinfo{title}{Describing groups}}.
\newblock {\sl \bibinfo{journal}{Bull. Symbolic Logic}}
  \bibinfo{volume}{13}(\bibinfo{number}{3}), pp. \bibinfo{pages}{305--339},
  \doi{10.2178/bsl/1186666149}.

\bibitemdeclare{inproceedings}{RaTh06}
\bibitem{RaTh06}
\bibinfo{author}{Alexander~Moshe \surnamestart Rabinovich\surnameend} \&
  \bibinfo{author}{Wolfgang \surnamestart Thomas\surnameend}
  (\bibinfo{year}{2006}): \emph{\bibinfo{title}{Decidable Theories of the
  Ordering of Natural Numbers with Unary Predicates}}.
\newblock In: {\sl \bibinfo{booktitle}{CSL}}, pp. \bibinfo{pages}{562--574}.
\newblock \urlprefix\url{http://dx.doi.org/10.1007/11874683_37}.

\bibitemdeclare{article}{Rubi08}
\bibitem{Rubi08}
\bibinfo{author}{Sasha \surnamestart Rubin\surnameend} (\bibinfo{year}{2008}):
  \emph{\bibinfo{title}{Automata Presenting Structures: A Survey of the Finite
  String Case}}.
\newblock {\sl \bibinfo{journal}{Bulletin of Symbolic Logic}}
  \bibinfo{volume}{14}(\bibinfo{number}{2}), pp. \bibinfo{pages}{169--209},
  \doi{10.2178/bsl/1208442827}.

\bibitemdeclare{article}{Tsan11}
\bibitem{Tsan11}
\bibinfo{author}{Todor \surnamestart Tsankov\surnameend}
  (\bibinfo{year}{2011}): \emph{\bibinfo{title}{The additive group of the
  rationals does not have an automatic presentation}}.
\newblock {\sl \bibinfo{journal}{J. Symbolic Logic}}
  \bibinfo{volume}{76}(\bibinfo{number}{4}), pp. \bibinfo{pages}{1341--1351},
  \doi{10.2178/jsl/1318338853}.

\end{thebibliography}
\bibliographystyle{eptcs}

\end{document}